\documentclass[a4paper,UKenglish,cleveref, autoref, thm-restate]{lipics-v2021}

 \usepackage{amsmath,amsfonts,amssymb,amsthm}
 \usepackage{mathbbol}
\usepackage{amsthm}
\usepackage{graphicx,color}
\usepackage{boxedminipage}
 \usepackage[linesnumbered, boxed, algosection,noline,noend]{algorithm2e}
 \usepackage{framed}
\usepackage{xspace}
\usepackage{todonotes}
\usepackage{xifthen}
 \usepackage{tabularx}
  \usepackage{capt-of}
\usepackage{thmtools}
 \usepackage{subcaption}
  \usepackage{booktabs}
\usepackage{calc}
\usepackage{dirtytalk}
\usepackage{bm}
\usepackage{enumitem}
 \usepackage{edtable}

\usepackage{mdframed}
\usepackage{multirow}
\usepackage{diagbox}

\usetikzlibrary{calc}
\usetikzlibrary{shapes}
\usetikzlibrary{arrows.meta}

\usepackage{mathrsfs}
 
\DeclareMathOperator{\operatorClassNP}{{\sf NP}}
\newcommand{\classNP}{\ensuremath{\operatorClassNP}}

\newcommand{\cost}{\mathrm{cost}}

\usepackage{array}
\newcolumntype{P}[1]{>{\centering\arraybackslash}p{#1}}



\newlength{\RoundedBoxWidth}
\newsavebox{\GrayRoundedBox}
\newenvironment{GrayBox}[1]%
   {\setlength{\RoundedBoxWidth}{.93\textwidth}
    \def\boxheading{#1}
    \begin{lrbox}{\GrayRoundedBox}
       \begin{minipage}{\RoundedBoxWidth}}%
   {   \end{minipage}
    \end{lrbox}
    \begin{center}
    \begin{tikzpicture}%
       \node(Text)[draw=black!20,fill=white,rounded corners,%
             inner sep=2ex,text width=\RoundedBoxWidth]%
             {\usebox{\GrayRoundedBox}};
        \coordinate(x) at (current bounding box.north west);
        \node [draw=white,rectangle,inner sep=3pt,anchor=north west,fill=white] 
        at ($(x)+(6pt,.75em)$) {\boxheading};
    \end{tikzpicture}
    \end{center}}     

\newenvironment{defproblemx}[2][]{\noindent\ignorespaces%
                                \FrameSep=6pt%
                                \parindent=0pt%
                \vspace*{-1.5em}
                \ifthenelse{\isempty{#1}}{%
                  \begin{GrayBox}{\textsc{#2}}%
                }{%
                  \begin{GrayBox}{\textsc{#2} parameterized by~{#1}}%
                }
                \begin{tabular*}{\textwidth}{@{\hspace{.1em}} >{\itshape} p{1.8cm} p{0.8\textwidth} @{}}%
            }{
                \end{tabular*}%
                \end{GrayBox}%
                \ignorespacesafterend
            }  

\title{
FPT Approximation for Fair Minimum-Load Clustering} 

\author{Sayan Bandyapadhyay}{Department of Informatics, University of Bergen, Norway}{Sayan.Bandyapadhyay@uib.no}{}{}

\author{Fedor V. Fomin}{Department of Informatics, University of Bergen, Norway}{Fedor.Fomin@uib.no}{https://orcid.org/0000-0003-1955-4612}{}
\author{Petr A. Golovach}{Department of Informatics, University of Bergen, Norway}{Petr.Golovach@uib.no }{https://orcid.org/0000-0002-2619-2990}{}

\author{Nidhi Purohit}{Department of Informatics, University of Bergen, Norway}{Nidhi.Purohit@uib.no}{}{}

\author{Kirill Simonov}{Department of Informatics, University of Bergen, Norway}{kirillsimonov@gmail.com}{}{}

\authorrunning{S. Bandyapadhyay, F.\,V. Fomin, 
P.\,A. Golovach,
N. Purohit, 
K. Simonov}
\Copyright{
Sayan Bandyapadhyay, Fedor V. Fomin, 
Petr A. Golovach, 
Nidhi Purohit, 
Kirill Simonov
}

\ccsdesc[100]{Theory of computation~Design and analysis of algorithms~Approximation algorithms analysis~Facility location and clustering}

\keywords{fair clustering, load balancing, parameterized approximation, minimum-load, $k$-median} 

\category{} 

\supplement{}

\funding{The research leading to these results has received funding from the Research Council of Norway via the project ```MULTIVAL" (grant no. 263317) and the European Research Council (ERC) via grant LOPPRE, reference 819416.}

\nolinenumbers 

\EventEditors{John Q. Open and Joan R. Access}
\EventNoEds{2}
\EventLongTitle{42nd Conference on Very Important Topics (CVIT 2016)}
\EventShortTitle{CVIT 2016}
\EventAcronym{CVIT}
\EventYear{2016}
\EventDate{December 24--27, 2016}
\EventLocation{Little Whinging, United Kingdom}
\EventLogo{}
\SeriesVolume{42}
\ArticleNo{23}

\begin{document}

\maketitle

\begin{abstract}
In this paper, we consider the Minimum-Load $k$-Clustering/Facility Location (MLkC) problem where we are given a set $P$ of $n$ points in a metric space that we have to cluster and an integer $k > 0$ that denotes the number of clusters. Additionally, we are given a set $F$ of cluster centers in the same metric space. The goal is to select a set $C\subseteq F$ of $k$ centers and assign each point in $P$ to a center in $C$, such that the maximum \emph{load} over all centers is minimized. Here the load of a center is the sum of the distances between it and the points assigned to it. 

Although clustering/facility location problems have a rich literature, the minimum-load objective is not studied substantially, and hence MLkC has remained a poorly understood problem. More interestingly, the problem is notoriously hard even in some special cases including the one in line metrics as shown by Ahmadian et al. [APPROX 2014, ACM Trans. Algorithms 2018]. They also show \textsf{APX}-hardness of the problem in the plane. On the other hand, the best known approximation factor for MLkC is $O(k)$, even in the plane. 

In this work, we study a fair version of MLkC inspired by the work of Chierichetti et al. [NeurIPS, 2017], which generalizes MLkC. Here the input points are colored by one of the $\ell$ colors denoting the group they belong to. MLkC is the special case with $\ell=1$. Considering this problem, we are able to obtain a $3$-approximation in $f(k,\ell)\cdot n^{O(1)}$ time. Also, our scheme leads to an improved $(1 + \epsilon)$-approximation in case of Euclidean norm, and in this case, the running time depends only polynomially on the dimension $d$. Our results imply the same approximations for MLkC with running time $f(k)\cdot n^{O(1)}$, achieving the first constant approximations for this problem in general and Euclidean metric spaces. Our work makes substantial progress in understanding the structure of the minimum-load objective, advancing the frontiers of knowledge about the  MLkC problem. 
\end{abstract}

\section{Introduction}
\label{sec:intro}
\emph{Clustering} is the task of partitioning a set of data items into a number of groups (or clusters) such that each group contains similar set of items. Typically, the similarity of the clusters is modeled by a proxy objective function, which one needs to optimize. Being a fundamental computational problem in nature, clustering has a host of diverse applications in computer science and other disciplines. Consequently, the problem has been studied with several different and possibly independent objectives. Some of these became notably popular, for example, $k$-means, $k$-median, and $k$-center \cite{gonzalez1985clustering,feder1988optimal,AryaGKMMP-SIAMJ04,KanungoMNPSW04}. In this paper, we consider an objective which is not studied substantially in the literature. In particular, we consider minimum-load clustering. Here we are given a set $P$ of points in a metric space that we have to cluster and an integer $k > 0$ that denotes the number of clusters. Additionally, we are given a set $F$ of cluster centers in the same metric space. The goal is to select a set $C\subseteq F$ of $k$ centers and assign each point in $P$ to a center in $C$, such that the maximum \emph{load} over all centers is minimized. Here the load of a center is the sum of the distances between it and the points assigned to it. That is, if $P'$ is the set of points assigned to a center $c$, then its load is $\sum_{p\in P'} d(c,p)$, where $d$ is the given metric. We formally refer to this problem as Minimum-Load $k$-Clustering (MLkC). MLkC can be used to model applications where the cost of serving the clients (or points) assigned 
to a facility (or center) is incurred by the facility, e.g., assigning jobs to the $k$ best servers from a pool of servers balancing their loads.   

Surprisingly, MLkC is \classNP-hard even if the solution set of centers $C$ is given, via a reduction from makespan-minimization \cite{AhmadianBFJSS18}. In fact, this assignment version of the problem can be shown to be \classNP-hard even in line metrics and for $k=2$, via a simple reduction from the Partition problem \cite{KORF1998181}. (In Partition, given a set of integers, the goal is to partition it into two subsets such that the difference between the sums of the integers in two subsets is minimized.) Moreover, Ahmadian et al.~\cite{AhmadianBFJSS18} proved that  the problem is strongly \classNP-hard in line metrics (points on a line) and \textsf{APX}-hard in the plane. On the positive side, an $O(k)$-approximation follows for this problem from any existing $O(1)$-approximation for $k$-median \cite{CharikarGTS-STOC99,AryaGKMMP-SIAMJ04,DBLP:journals/jacm/JainV01,DBLP:journals/siamcomp/LiS16,DBLP:journals/talg/ByrkaPRST17}. This is true, as $k$-median minimizes the sum of the loads of the centers. Also, constant-approximations are known for MLkC in some special cases, e.g., in star metrics and line metrics. Beyond these special cases, obtaining better than $O(k)$-approximation in polynomial time remained a notoriously hard question, even in the plane. Indeed, as explicitly pointed out by Ahmadian et al.~\cite{AhmadianBFJSS18}, MLkC is resilient to attack by the standard approximation techniques including LP rounding and local search, which has been fairly successful in obtaining good approximation algorithms for other clustering problems. Given these difficulties, we investigate whether it is possible to obtain $O(1)$-approximation for MLkC if we allow time $f(k)\cdot n^{O(1)}$ instead of only $n^{O(1)}$, for some function $f(.)$ independent of the input size $n$. Indeed, we study a much more general \emph{fair} version of the problem. 

Fair clustering was introduced by Chierichetti et al.~\cite{chierichetti2017fair} with the goal of removing inherent biases from the regular clustering models. In this setting, we also have a sensitive or protected feature of the data points, e.g., gender or race. The goal is to obtain a clustering where the fraction of points from a
traditionally underrepresented group (w.r.t. the protected feature) in every cluster is approximately equal to the fraction of points from
this group in the whole dataset. For simplicity, they assumed that  the protected feature can take only two values and designed fair $k$-center and $k$-median clustering algorithms in this setting. In particular, here one is given two sets of points $R$ and $B$ of color red and blue, respectively, and a balance parameter $t\in [0,1]$. The objective is to find a clustering such that in every cluster $O$, the ratio between the number of red points and the number of blue points is at least $t$ and at most $1/t$, i.e., $t \le \frac{|O\cap R|}{|O\cap B|} \le 1/t$. Subsequently, R\"{o}sner and Schmidt \cite{rosner2018privacy} considered a general model where the protected feature can take any number of values and designed fair clustering algorithms for $k$-center objective. Later, Bercea et al. \cite{bercea2019cost} and Bera et al. \cite{bera2019fair} independently considered a fair clustering model that generalizes the models in both \cite{chierichetti2017fair} and \cite{rosner2018privacy}. In this model, we are given a partition $\{P_1$, $P_2$, \ldots, $P_\ell\}$ of the input point set $P$ and balance parameters $0\le \beta_i\le \alpha_i\le 1$ for each group $1\le i\le \ell$. Then a clustering is called $(\alpha,\beta)$-\emph{fair} if the fraction of points
from each group $i$ in every cluster is at least $\beta_i$ and at most $\alpha_i$. In this paper, we study the $(\alpha,\beta)$-Fair Minimum-Load $k$-Clustering (FMLkC) problem, where the goal is to compute an $(\alpha,\beta)$-fair clustering that minimizes the maximum load. (For a formal definition, please see Section \ref{sec:prelims}.) We note that the only clustering objectives considered in all the above mentioned  works on fair clustering are $k$-means, $k$-median and $k$-center. To the best of our knowledge, fair clustering was not studied with the  minimum-load objective before our work.   

\subsection{Our Results and Techniques}
Considering the FMLkC problem in general and Euclidean metric spaces we obtain the following results.

\begin{theorem}[\textbf{Informal}]
    There is a $3$-approximation algorithm for $(\alpha,\beta)$-Fair Minimum-Load $k$-Clustering in general metric spaces that runs in time $2^{\tilde{O}(k \ell^2)} n^{O(1)}$. For $d$-dimensional Euclidean spaces, there is a $(1 + \epsilon)$-approximation algorithm for $(\alpha,\beta)$-Fair Minimum-Load $k$-Clustering with running time
    $2^{\tilde{O}(k \ell^2/\epsilon^{O(1)})} n^{O(1)}$. 
\end{theorem}

In the above theorem, the $\tilde{O}()$ notation hides logarithmic factors. Note that all the running times are \emph{fixed-parameter tractable} (FPT) \cite{cygan2015parameterized}  in $k$ and $\ell$ for  constant $\epsilon$. Moreover, our results imply the same approximations for Minimum-Load $k$-Clustering with running times FPT in only $k$, achieving the first constant approximations for this problem in general and Euclidean metric spaces. Note that in the Euclidean case, the running time depends only polynomially on the dimension $d$. Recall that no better than $O(k)$-approximation was known before even in the plane, and this version is known to be \textsf{APX}-hard. Also, the reduction  mentioned before from Partition eliminates the existence of an exact algorithm in time $f(k)\cdot n^{O(1)}$, unless \textsf{P} $\ne$  \textsf{NP}, as MLkC in line metrics is already \classNP-hard when $k=2$. In this sense, our FPT $(1+\epsilon)$-approximation for Euclidean spaces is tight and the best possible.  

Our results are motivated by the recent FPT approximation results for constrained clustering with popular $k$-median and $k$-means objectives \cite{Cohen-AddadL19,abs-2007-10137}. However, these results are based on coreset construction. A coreset is a summary of the original dataset from which it is possible to retrieve a near-optimal clustering. Their main contribution is to show that it is possible to obtain coresets of size polynomial in $k, \log n$ and $d$. Alternatively, the input can be compressed to an almost equivalent instance of size $poly(kd\log n)$. Then one can enumerate all possible $k$-tuples of centers in FPT time using the coreset and output the $k$-tuple having the minimum clustering cost.  This yields a $(1+\epsilon)$-approximation for Euclidean spaces and a slightly larger $3$-approximation for general metric spaces due to some technical reasons. However, such a small-sized coreset is not known for our problems. Instead, we adapt approaches from \cite{ding2020unified,GoyalJK20,Bhattacharya2018} used for directly obtaining FPT approximations for constrained $k$-median and $k$-means clustering. We note that these schemes were known only in the special Euclidean case until recently \cite{GoyalJK20}. All these schemes produce in FPT (in $k$) time a list of $k$-tuples of centers, such that at least one such $k$-tuple is a near-optimal set of centers. Using the rather tedious  similarity of the $k$-median and the minimum-load objectives, we show these approaches can be adapted for our problems as well. However, given such a $k$-tuple of centers, assigning the points to the best centers or finding the optimal clustering, in our case is still \classNP-hard. Nevertheless, we give a Mixed-Integer Linear Programming (MILP) based $(1+\epsilon)$-approximation for this assignment problem that runs in time FPT in $k$ and $\ell$ (in $k$ only for MLkC). Our MILP is partly motivated by the fair $k$-median MILP \cite{abs-2007-10137}. However, our MILP and its rounding are much more non-trivial compared to that for fair $k$-median, especially due to the difference in the objectives. For example, if we forget about the fairness constraints, in that case the assignment algorithm for $k$-median is trivial. Assign each point to its closest center. However, even in this case the assignment problem for MLkC is \classNP-hard. Also, no near-optimal assignment scheme was known in the literature (a $2$-approximate assignment scheme follows from the generalized assignment problem (GAP) \cite{DBLP:journals/mp/ShmoysT93}). Thus, in this case we give a novel $(1+\epsilon)$-approximate assignment scheme. In this case, we do not need MILP -- rounding of an LP is sufficient to obtain the desired assignment.  All these schemes applied together help us achieve the desired FPT approximations.

\subsection{Related Work}
Even et al.~\cite{EvenGKRS03} and Arkin et al.~\cite{ArkinHL06} studied the MLkC problem under the name \emph{min-max star cover}, where $F=P$. In this setting, MLkC can be viewed as a weighted covering problem where the task is to cover the nodes of a graph by stars. Both works obtain \emph{bicriteria} approximation for this problem where the solution returned has near-optimal load, but uses more than $k$ centers.  Ahmadian et al.~\cite{AhmadianBFJSS18} studied several special cases of the MLkC problem (under the name Minimum-Load $k$-Facility Location\footnote{MLkC can also be viewed as a facility location problem with zero facility opening costs where we can still open only $k$ facilities}). They fully resolved the status  of the MLkC problem in line metrics. On the one hand, they designed a PTAS based on dynamic programming. On the other hand, they proved that this version is strongly \classNP-hard. They also designed a quasi-PTAS in tree metrics. Moreover, they studied a variant of the problem with client demands in star metrics. 

The notion of fair clustering introduced by Chierichetti et al.~\cite{chierichetti2017fair} has been studied extensively in the literature. For $k$-center objective, there are several polynomial-time true $O(1)$-approximations \cite{rosner2018privacy,bercea2019cost}. For $k$-median and $k$-means objectives, polynomial-time $O(1)$-approximations are designed by violating the fairness constraints by an additive factor \cite{bercea2019cost,bera2019fair}. On the other hand, it is possible to obtain true $O(1)$-approximations for these two objectives if one is allowed to use $f(k,\ell)\cdot n^{O(1)}$ time \cite{abs-2007-10137}. Clustering problems have been studied under several other notions of fairness, e.g., see \cite{ahmadian2019clustering,chen2019proportionally, kleindessner2019fair, kleindessner2019guarantees,GhadiriSV21,abs-2103-02512,AbbasiBV21}.    

\subparagraph*{Organization.} We define some useful notations and our problem formally in Section \ref{sec:prelims}. Then we describe the assignment algorithm for the FMLkC problem in Section \ref{sec:assignment}. Finally, in Section \ref{sec:algo}, we describe the full algorithms for FMLkC in details.

\section{Preliminaries}
\label{sec:prelims}

We are given a set $P$ of points in a metric space $(\mathcal{X}, d(\cdot, \cdot))$, that we have to cluster. We are also given a set $F$ of cluster centers in the same metric space. We note that $P$ and $F$ are not-necessarily disjoint, and in fact, $P$ may be equal to $F$. In the Euclidean version of a clustering problem, $P\subseteq \mathbb{R}^d$, $F= \mathbb{R}^d$ and $d(\cdot, \cdot)$ is the Euclidean distance.\footnote{Due to the lack of  better notations, we denote the dimension by $d$ and distance function by $d(\cdot, \cdot)$} In the metric version, we assume that $F$ is finite. Thus, strictly speaking, the Euclidean version is not a special case of the metric version. 
In the metric version, we denote $|P\cup F|$ by $n$ and in the Euclidean version, $|P|$ by $n$.  
For any integer $t\ge 1$, we denote the set $\{1,2,\ldots,t\}$ by $[t]$. 

For a partition $\mathbb{O} = \{O_1, \ldots, O_k\}$ of $P$ and a set of $k$ cluster centers $C = \{c_1, \ldots, c_k\} \subset F$, we say that $\mathbb{O}$ is a \emph{clustering} of $P$ with the centers $c_1$, \ldots, $c_k$. We say that the \emph{minimum-load cost} of this clustering (also, simply \emph{cost of clustering}) is $\max_{i \in [k]} \cost_{c_i}(O_i)$. Here $\cost_{c_i}(O_i)$ denotes the sum-of-distances cost of the cluster $O_i$ with the center $c_i$, which is
$\cost_{c_i}(O_i) = \sum_{x \in O_i} d(x, c_i).$
We use the following notation to denote the cost of clustering w.r.t. the set of centers $C$ up to a permutation of the clusters,
\[\cost_{C}(\mathbb{O}) = \min_{i_1, \ldots, i_k} \max_{j \in [k]} \cost_{c_j}(O_{i_j}),\]
where $i_1$, \ldots, $i_k$ is a permutation of $[k]$.
We also denote by $\cost(O_i)$ the cost of a cluster $O_i$ with the optimal choice of a center, that is,
$\cost(O_i) = \min_{c \in F} \cost_c (O_i)$,
and by $\cost(\mathbb{O})$ the optimal cost of clustering $\mathbb{O}$, 
\[\cost(\mathbb{O}) = \min_{\substack{C \subset F\\|C| = k}} \cost_C (\mathbb{O}).\]
Alternatively, a clustering with centers in $C \subset F$ can be defined as an assignment $\varphi: P \to C$. The assignment $\varphi$ then corresponds to a clustering $\{\varphi^{-1}(c)\}_{c \in C}$, and we say that the cost of the assignment $\varphi$ is $\cost(\varphi)=\max_{c\in C} \sum_{x\in P: \varphi(x) = c} d(x, c)$.

Now we define the main problem of our interest, where the goal is to find the minimum-cost clustering that satisfies the fairness constraints.
\begin{definition}
    \label{definition:fair}
    In the $(\alpha,\beta)$-Fair Minimum-Load $k$-Clustering (FMLkC) problem, we are given a partition $\{P_1$, $P_2$, \ldots, $P_\ell\}$ of $P$. We are also given an integer $k>0$ and two \emph{fairness vectors} $\alpha, \beta \in [0, 1]^\ell$, $\alpha = (\alpha_1, \ldots, \alpha_\ell)$, $\beta = (\beta_1, \ldots, \beta_\ell)$. The objective is to select a set of at most $k$ centers $C \subset F$ and an assignment $\varphi: P \to C$ such that $\varphi$ satisfies the following \emph{fairness constraints}:
    \begin{gather*}
        \left|\{x \in P_i : \varphi(x) = c\}\right| \le \alpha_i \cdot \left|\{x \in P : \varphi(x) = c\}\right|, \quad \forall c \in C, \forall i \in [\ell],\\
        \left|\{x \in P_i : \varphi(x) = c\}\right| \ge \beta_i \cdot \left|\{x \in P : \varphi(x) = c\}\right|, \quad \forall c \in C, \forall i \in [\ell],
    \end{gather*}
    and $\cost(\varphi)$ is minimized among all such assignments.
\end{definition}

Minimum-Load $k$-Clustering (MLkC) is a restricted case of FMLkC with $\ell=1$, and hence there is no fairness constraints involved in this case. The $(\alpha,\beta)$-Fair $k$-median problem is defined identically except there the cost is $\cost(\varphi)=\sum_{c\in C} \sum_{x\in P: \varphi(x) = c} d(x, c).$

 \section{Assignment Problem for FMLkC}
\label{sec:assignment}

In the $(\alpha,\beta)$-fair assignment problem, we are additionally given a set of centers $C\subset F$ and the goal is to find an assignment $\varphi: P \to C$ such that $\varphi$ satisfies the fairness constraints and $\cost(\varphi)=\max_{c\in C} \sum_{x\in P: \varphi(x) = c} d(x, c)$ is minimized.

We refer to an assignment as a fair assignment if it satisfies the fairness constraints. Also, we denote the optimal cost of $(\alpha,\beta)$-fair assignment by OPT. In this section, for any $\epsilon > 0$, we give a $(1+\epsilon)$-approximation for this problem in $f(k,\ell, \epsilon)\cdot n^{O(1)}$ time for some computable function $f$. In particular, we solve a budgeted version of the problem where we are also given a budget $B$ and the goal is to decide whether there is a fair assignment of cost at most $B$. 

\begin{lemma}
Suppose there is an algorithm $\mathcal{A}$ that given an instance of budgeted $(\alpha,\beta)$-fair assignment and any $\epsilon > 0$, in $T(n,k,\ell,\epsilon)$ time, either returns a feasible assignment of cost at most $(1+\epsilon)B$, or correctly detects that there is no feasible assignment with budget $B$. Then for any $\epsilon > 0$, one can obtain a $(1+\epsilon)$-approximation for $(\alpha,\beta)$-fair assignment in $(k\ell)^{O(k\ell)} n^{O(1)}+O_{\epsilon}(\log k)\cdot T(n,k,\ell,\epsilon/3)$ time\footnote{$O_{\epsilon}()$ notation hides $O(1/\epsilon)$ factor.}. 
\end{lemma}

\begin{proof}
The idea is to first find a range where OPT belongs and then apply $\mathcal{A}$ with budget within this range to find a feasible assignment. Given an instance $I$ of $(\alpha,\beta)$-fair assignment, first we use an algorithm (Theorem 8.2, \cite{abs-2007-10137}) to compute a fair assignment of the points to the centers of $C$ that minimizes the $(\alpha,\beta)$-fair $k$-median cost. This algorithm runs in time $(k\ell)^{O(k\ell)} n^{O(1)}$. Let $D$ be the computed $(\alpha,\beta)$-fair $k$-median cost returned by the algorithm. Then $D \le k\cdot $OPT, as the optimal cost of $(\alpha,\beta)$-fair assignment is at least $1/k$ fraction of the optimal $(\alpha,\beta)$-fair $k$-median cost. Also, $OPT \le D$, as optimal $(\alpha,\beta)$-fair assignment cost is at most the optimal $(\alpha,\beta)$-fair $k$-median cost. Hence $D/k\le$ OPT $\le D$. 

Let $\epsilon' = \epsilon/3$ and $m$ be the maximum $t$ such that $(1+\epsilon')^t \le D/k$. Also, let $M$ be the minimum $t$ such that $D \le (1+\epsilon')^t$. Thus  $(1+\epsilon')^m\le $ OPT $\le (1+\epsilon')^M$. We run the algorithm $\mathcal{A}$ setting $\epsilon$ to be $\epsilon'$ for budget $B=(1+\epsilon')^i$ where $m\le i\le M$, and terminate it the first time it returns a feasible assignment for some budget $B$. Let $B'$ be the budget for which this algorithm returns a feasible assignment. Then $B' \le (1+\epsilon')$ OPT, as any instance with budget $B \ge $ OPT is a yes-instance, and for such a $B$, $\mathcal{A}$ returns a feasible assignment of cost at most $(1+\epsilon)B$. Hence, the cost of the assignment returned by $\mathcal{A}$ with budget $B'$ is at most  $(1+\epsilon')^2$ OPT $\le (1+\epsilon)$ OPT. As the algorithm $\mathcal{A}$ can be used at most $O_{\epsilon}(\log (M-m+1))=O_{\epsilon}(\log (D/(D/k)))=O_{\epsilon}(\log k)$ times, the whole algorithm runs in time $(k\ell)^{O(k\ell)} n^{O(1)}+O_{\epsilon}(\log k)\cdot  T(n,k,\ell,\epsilon/3)$.   
\end{proof}

In the following, we design an LP rounding based algorithm for budgeted $(\alpha,\beta)$-fair assignment with the properties required in the above lemma. Moreover, this algorithm runs in time $k^{O(k\ell)} {\ell}^{O((k\ell^2/\epsilon)\log (\ell/\epsilon))} n^{O(1)}$. Hence, we obtain the following theorem. 

\begin{theorem}\label{thm:fairassgn}
For any $\epsilon > 0$, a $(1+\epsilon)$-approximation for $(\alpha,\beta)$-fair assignment can be obtained in time $k^{O(k\ell)} {\ell}^{O((k\ell^2/\epsilon)\log (\ell/\epsilon))} n^{O(1)}$.
\end{theorem}

Next, we design the algorithm for the budgeted version of $(\alpha,\beta)$-fair assignment. Recall that in the budgeted version, we are given an instance $I$ containing $\ell$ disjoint groups $\{P_i\}$ of $P=\{p_1,\ldots,p_n\}$, a set of $k$ centers $C=\{c_1,\ldots,c_k\}$ and the budget $B$. Our algorithm first rounds each distance to a power of $(1+\epsilon)$. Fix any center $c_i\in C$. We partition the points in $P$ into a number of classes based on their distances from $c_i$. For all $p\in P$, let $\hat{d}(p,c_i)=(1+\epsilon)^t\epsilon^2 B$ where $t=\lceil\log_{1+\epsilon} d(p,c_i)/(\epsilon^2 B)\rceil$.  Let $d_t=(1+\epsilon)^t\epsilon^2 B$. We refer to the points $p$ with distance $\hat{d}(p,c_i)=d_t$ as the \emph{distance class} $t$ with respect to (w.r.t.) $c_i$, which is denoted by $S_{it}$. 

\begin{observation}
For all $p\in P, c\in C$, ${d}(p,c) \le \hat{d}(p,c) \le (1+\epsilon)\cdot {d}(p,c)$. 
\end{observation}

Let $I'$ be the new instance of budgeted $(\alpha,\beta)$-fair assignment with the modified distance $\hat{d}$. As $\hat{d}$ is  obtained by scaling $d$ by at most a factor of $(1+\epsilon) $, we have the following observation. 

\begin{observation}
If there is a feasible assignment for $I$ with budget $B$, then there is a feasible assignment for $I'$ with budget $(1+\epsilon) B$. Also, if there is a feasible assignment for $I'$ with budget $(1+\epsilon) B$, then there is a feasible assignment for $I$ with budget $(1+\epsilon) B$.
\end{observation}

By the above observation, it is sufficient to consider $\hat{d}$ instead of $d$ for the purpose of computing an assignment of cost at most $(1+\epsilon) B$. Henceforth, by distance we mean $\hat{d}$. 

Denote by $\varphi^*$ a feasible assignment for $I'$ of cost at most $(1+\epsilon) B$ (if any). We define a point $p \in P$ to be \emph{costly} w.r.t. a center $c\in C$ if $\hat{d}(p,c) \ge \epsilon^2 B$. Otherwise, we define the point to be \emph{cheap} w.r.t. $c$. Note that the number of costly points that can be assigned to each center in $\varphi^*$ is at most $(1+\epsilon)/\epsilon^2\le 2/\epsilon^2$. The next observation follows from the fact that $d_t=(1+\epsilon)^t\epsilon^2 B$. 

\begin{observation}
For any point $p$ and center $c\in C$ with $\hat{d}(p,c)=d_t$, $t < 0$ if $p$ is cheap w.r.t. $c$, and $t \ge 0$ if $p$ is costly w.r.t. $c$. 
\end{observation}

Now, as we are shooting for an assignment of cost at most $(1+\epsilon) B$, we can discard all the distances $\hat{d}(p,c)$ larger than $(1+\epsilon) B$, i.e., we can assume that such a $p$ will never be assigned to $c$. Without loss of generality, we assume that all the distances we have are bounded by $(1+\epsilon) B$. Let $\Delta$ be the maximum $t$ such that there are $p\in P$ and $c\in C$ with $\hat{d}(p,c)=d_t$ for a costly point $p$ w.r.t. $c$. By our previous assumption, $\hat{d}(p,c)\le (1+\epsilon) B$. Thus $\Delta\le \lceil\log_{1+\epsilon} ((1+\epsilon)B/(\epsilon^2 B))\rceil =O((1/\epsilon)\log (1/\epsilon))$, which is a constant. For $1\le i\le k$, $0 \le t \le \Delta$ and $1\le g\le \ell$, let $z_{i,t,g}$ be the number of costly points $p\in P_g$  assigned to $c_i$ in $\varphi^*$ such that $\hat{d}(p,c_i)=d_t$. Note that each $z_{i,t,g} \le 2/\epsilon^2$, as the total number of costly points assigned to a center is at most $2/\epsilon^2$. Thus the total number of distinct choices for these variables is $(2/\epsilon^2)^{k \ell\Delta}=(1/\epsilon)^{O((k\ell/\epsilon)\log (1/\epsilon))}.$ 

\begin{observation}\label{obs:z-values}
There are $(1/\epsilon)^{O((k\ell/\epsilon)\log (1/\epsilon))}$ distinct choices for the variables $\{z_{i,t, g}: 1\le i\le k, 0 \le t \le \Delta, 1\le g\le \ell\}$. 
\end{observation}

As we can probe all such possible choices, we assume that we know the exact values of these variables in $\varphi^*$. Next, we describe a Mixed-Integer Linear Program (MILP) for the budgeted version of the problem, which is partly motivated by the $(\alpha,\beta)$-fair $k$-median MILP \cite{abs-2007-10137}. For every point $p_j$ and center $c_i$, we have a fractional variable $x_{ij}$ denoting the extent up to which $p_j$ is assigned to $c_i$. For every center $c_i$ and group $g \in \{1, \ldots, \ell\}$, we have an integral variable $y_{gi}$ denoting the ``weight'' of the points assigned to $c_i$ from group $g$. The constraints of the MILP are described as follows. The Constraint \ref{eq:ilp1} ensures that each point is assigned to the centers up to an extent of 1. Constraint \ref{eq:ilp2} ensures that the weight assigned from each group $g$ to each center $c_i$ is exactly $y_{gi}$. Constraints \ref{eq:ilp3} and \ref{eq:ilp4} are fairness constraints. Constraint \ref{eq:ilp5} ensures that the weight of costly points from each class $t$ and group $g$ assigned to each $c_i$ is exactly same as the guessed value $z_{i,t,g}$. Constraint \ref{eq:ilp6} ensures that the total load assigned to each $c_i$ is bounded by $(1+\epsilon)B$. The first and the second expressions on the left hand side of this constraint are corresponding to costly and cheap points, respectively.       
\begin{align}
    &\sum_{1 \le i \le k} x_{ij} = 1  &\forall j \in [n] \label{eq:ilp1}\\
    &\sum_{j \in [n] : p_j \in P_g} x_{ij} = y_{gi}  &\forall i \in [k],\  \forall g \in [\ell], \label{eq:ilp2}\\
    &y_{gi}\ge \beta_g\sum_{j \in [n]} x_{ij} 
    &\forall i \in [k],\ \forall g \in [\ell]\label{eq:ilp3}\\
    &y_{gi}\le \alpha_g\sum_{j \in [n]} x_{ij}
    &\forall i \in [k],\ \forall g \in [\ell]\label{eq:ilp4}\\
    &\sum_{p_j\in P_g\cap S_{it}} x_{ij} = z_{i,t,g} &\forall i \in [k],\ \forall t \in \{0, \ldots, \Delta\},\ \forall g \in [\ell]\label{eq:ilp5}\\
    &\sum_{g=1}^{\ell}\sum_{0 \le t \le \Delta} d_t z_{i,t,g}+\sum_{t < 0} d_t \sum_{p_j\in S_{it}} x_{ij} \le (1+\epsilon)B &\forall i \in [k] \label{eq:ilp6}\\
    &x_{ij} \ge 0  &\forall j \in [n],\  \forall i \in [k], \label{eq:ilp7}\\
    &y_{gi} \in \mathbb{Z}_{\ge 0}  &\forall i \in [k],\  \forall g \in [\ell].
\end{align}

Let us denote the above MILP by Fair-LP. A solution to Fair-LP is denoted by $(x,y)$. We note that the assignment $\varphi^*$ induces a feasible solution to Fair-LP. We use the following popular and celebrated result to solve this MILP. 

\begin{proposition}[\cite{Lenstra1983}, \cite{Kannan1987}, \cite{Frank1987}]
    \label{proposition:MILP}
    An MILP with $K$ integral variables and encoding size $L$, can be solved in time $K^{O(K)} L^{O(1)}$.
\end{proposition}

As Fair-LP has $k\ell$ integral variables $\{y_{gi}\}$ and polynomial encoding size, it can be solved in $(k\ell)^{O(k\ell)} n^{O(1)}$ time. If this algorithm outputs that there is no feasible solution to Fair-LP, we conclude that $I$ is a no-instance. Otherwise, let $(x^*,y^*)$ denote the feasible solution returned by this algorithm. Note that although the $y^*$ values are integral, $x^*$ values can very well be fractional. Next, we show how to round these variables to obtain an integral solution to Fair-LP such that the load of every center is increased by an additive factor of $O(\epsilon \ell B)$ compared to its load in $(x^*,y^*)$.  

Fix any group $g$. First, we show how to round the variables corresponding to the points of $P_g$. For this purpose, we construct a network $G_N=(V_N,E_N)$ with source $S$ and sink $T$ (see Figure \ref{fig:net}). For each point $p_j\in P_g$, there is a node $v_j$ in $V_N$. For each distance class $t$ of every center $c_i$, there is a node $w_{it}$. Also, for each center $c_i$, there is a node $u_i$. For each $v_j\in V_N$, there is an arc $(S,v_j)$ of capacity 1. For each $p_j\in P_g$ and center $c_i$, there is an arc $(v_j,w_{it})$ of capacity 1 where $t$ is the index such that $p_j \in S_{it}$, i.e, $\hat{d}(p_j,c_i)=d_t$. For center $c_i$ and distance class $t$, let $\lambda_{it}^g=\sum_{p_j\in P_g\cap S_{it}} x_{ij}^*$, i.e, the weight assigned from $P_g\cap S_{it}$ to $c_i$. For each node $w_{it}$, there is an arc $(w_{it},u_i)$ of capacity $\lceil \lambda_{it}^g\rceil$. Lastly, for each center $c_i$, there is an arc $(u_i,T)$ of capacity $y_{gi}^*$. 

Note that for each center $c_i$, the number of distance classes is at most the number of points $n$. Hence, the size of $V_N$ is a polynomial in $n$. Also, note that the solution $(x^*,y^*)$ projected on the points of $P_g$ induces a feasible fractional solution for the problem of computing a flow of value $|P_g|$ in $G_N$. 

\begin{observation}
 The network $G_N$ has a fractional flow of value $|P_g|$. 
\end{observation}

As all the capacities of the arcs are integral, by integrality of flow, there exists an integral feasible flow in $G_N$ of value $|P_g|$. We compute such a flow $f$ by using any polynomial time flow computation algorithm. This flow solution $f$ naturally gives us an integral assignment $\varphi_f$ of the points in $P_g$ to the centers in $C$. 

\begin{figure}[!ht]
		\begin{center}
			\includegraphics[width=.8\textwidth]{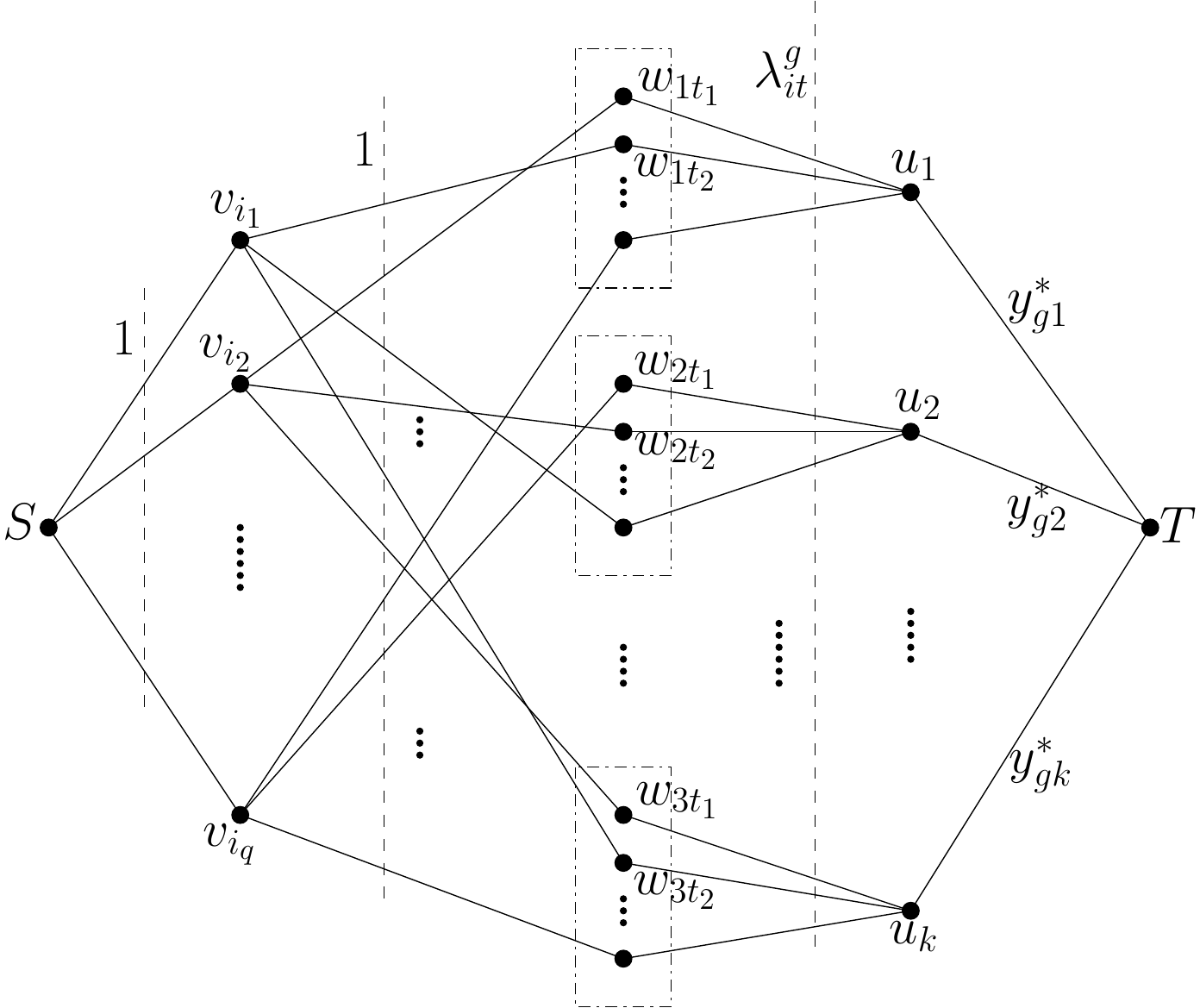}
			\caption{Figure showing the network $G_N$ constructed using the solution $(x^*,y^*)$.}
			\label{fig:net}
		\end{center}
\end{figure}

\begin{observation}\label{obs:saturatedarcs}
The number of points assigned to each center $c_i\in C$ via $\varphi_f$ is exactly $y_{gi}^*$. 
\end{observation}

\begin{proof}
Due to the capacity constraints of the arcs $\{(u_i,T)\}$, the number of points assigned to $c_i$ must be at most $y_{gi}^*$. Also, by definition, $\sum_{i=1}^{k} y_{gi}^*=|P_g|$. As $f$ has value $|P_g|$, the capacity of the arcs $\{(u_i,T)\}$ must be saturated, which completes the proof. 
\end{proof}

Next, we analyse the load of $P_g$ assigned to each center via $\varphi_f$. 

\begin{lemma}\label{lem:singlegrpcost}
    For each center $c_i\in C$, $\sum_{p_j\in P_g: \varphi_f(p_j)=c_i} \hat{d}(p_j,c_i) \le \sum_t d_t\lambda_{it}^g+O(\epsilon) B$.  
\end{lemma}

\begin{proof}
    Consider the arcs $\{(w_{it},u_i)\}$. The maximum flow corresponding to these arcs is bounded by the sum of the capacities $\sum_t \lceil \lambda_{it}^g\rceil$. Note that for every $0\le t\le \Delta$, $\lambda_{it}^g=z_{i,t,g}$, which is an integer. Now, 

\begin{align}
    \sum_{t < 0} d_t =\sum_{t < 0} (1+\epsilon)^t\epsilon^2 B < \epsilon^2 B((1+\epsilon)/\epsilon)=O(\epsilon) B \label{eq:sumofcheapdistances}
\end{align}
Hence, 
\begin{align*}
\sum_{t < 0} d_t\lceil \lambda_{it}^g\rceil \le \sum_{t < 0} d_t(\lambda_{it}^g+1)\le \sum_{t < 0}  d_t\lambda_{it}^g+\sum_{t < 0} d_t{=}_{[\ref{eq:sumofcheapdistances}]}\sum_{t < 0}  d_t\lambda_{it}^g+O(\epsilon) B
\end{align*}
It follows that, 
\begin{align*}
    \sum_{p_j\in P_g: \varphi_f(p_j)=c_i} \hat{d}(p_j,c_i)  & \le \sum_t  d_t\lceil \lambda_{it}^g\rceil\\
    &= \sum_{t \ge 0} d_t\lambda_{it}^g+ \sum_{t < 0} d_t\lceil \lambda_{it}^g\rceil\\
    &= \sum_{t \ge 0} d_t\lambda_{it}^g+\sum_{t < 0}  d_t\lambda_{it}^g+O(\epsilon) B\\
    &=\sum_t d_t\lambda_{it}^g+O(\epsilon) B
\end{align*}
\end{proof}

We repeat the above rounding process for all groups $g$. We combine the assignment functions $\varphi_f$ corresponding to the $\ell$ disjoint  groups to obtain a single assignment for the points in $P$. For simplicity, we also refer to this combined assignment as $\varphi_f$. By Observation \ref{obs:saturatedarcs}, $\varphi_f$ is feasible, as for each center $c_i$ and each group $g$, the weight of the points in $P_g$ assigned to $c_i$ is exactly $y_{gi}^*$ as in $(x^*,y^*)$. Next, we analyze the total load of each center. 

\begin{lemma}
    For each center $c_i\in C$, $\sum_{p_j: \varphi_f(p_j)=c_i} \hat{d}(p_j,c_i) \le  (1+O(\epsilon \ell)) B$
\end{lemma}

\begin{proof}
    \begin{align*}
       & \sum_{p_j: \varphi_f(p_j)=c_i} \hat{d}(p_j,c_i) 
        = \sum_{g=1}^{\ell} \sum_{p_j\in P_g: \varphi_f(p_j)=c_i} \hat{d}(p_j,c_i)\\
        &= \sum_{g=1}^{\ell} \Big(\sum_t d_t\lambda_{it}^g+O(\epsilon) B\Big) & (\text{By Lemma } \ref{lem:singlegrpcost})\\
        &= \sum_{g=1}^{\ell} \Big(\sum_{t \ge 0} d_t\lambda_{it}^g+\sum_{t < 0} d_t\lambda_{it}^g\Big) + O(\epsilon\ell) B\\
        &= \sum_{g=1}^{\ell} \Big(\sum_{t \ge 0} d_t z_{i,t,g}+\sum_{t < 0} d_t(\sum_{p_j\in P_g\cap S_{it}} x_{ij}^*)\Big)+ O(\epsilon\ell) B\\
        &= \bigg(\sum_{g=1}^{\ell} \sum_{t \ge 0} d_t z_{i,t,g}+\sum_{t < 0} d_t\sum_{p_j\in S_{it}} x_{ij}^*\bigg)+O(\epsilon\ell) B\\
        &\le (1+\epsilon) B+  O(\epsilon\ell) B & (\text{By Constraint } \ref{eq:ilp6} \text{ of Fair-LP})\\
        &=(1+O(\epsilon \ell)) B
    \end{align*}
\end{proof}

By scaling $\epsilon$ down by a factor of $\Omega(\ell)$, we obtain the desired approximate bound. Thus, by Observation \ref{obs:z-values}, it follows that the number of distinct possible choices of the $Z_{i,t,g}$ values is
$(\ell/\epsilon)^{O((k\ell^2/\epsilon)\log (\ell/\epsilon))}$. For each such choice, solving the MILP and rounding takes $(k\ell)^{O(k\ell)} n^{O(1)}$ time. Thus, the algorithm for solving the budgeted version runs in time $k^{O(k\ell)} {\ell}^{O((k\ell^2/\epsilon)\log (\ell/\epsilon))} n^{O(1)}$. The following lemma completes the proof of Theorem \ref{thm:fairassgn}. 

\begin{lemma}
The above MILP based algorithm for budgeted $(\alpha,\beta)$-fair assignment, in time $k^{O(k\ell)} {\ell}^{O((k\ell^2/\epsilon)\log (\ell/\epsilon))} n^{O(1)}$, either returns a feasible assignment of budget at most $(1+\epsilon)B$, or correctly detects that there is no feasible assignment of budget $B$.  
\end{lemma}

 \section{Approximation Algorithms for FMLkC}
\label{sec:algo}

In this section, we describe the FPT approximation algorithms for the FMLkC problem, both in the general metric case and in the Euclidean case.
In the general metric case, we aim for a $(3 + \epsilon)$-approximation, and in the Euclidean case for a $(1 + \epsilon)$ approximation, for a given $0 < \epsilon < 1$.
Essentially, we obtain these algorithms as a combination of our assignment algorithm presented before, and known generic results for constrained clustering problems that follow the framework of Ding and Xu~\cite{ding2020unified}. Our general metric algorithm employs the result of Goyal, Jaiswal and Kumar~\cite{GoyalJK20}, and in the Euclidean case we use the result of Bhattacharya, Jaiswal and Kumar~\cite{Bhattacharya2018}. Both of these provide algorithms that in FPT time produce a reasonably short list of candidate sets of $k$ centers, such that for each possible clustering of the input points one of the sets in the list provides the desired approximation, with good probability.
Note that the results mentioned above are stated in fact for the $k$-median objective, and not the minimum-load clustering that we study in this work. However, by tweaking the error guarantees in the respective proofs we can show that these results hold in the minimum-load setting as well. Next, we present these in detail.

We start with the Euclidean case and show the following analogue of Theorem~1 in~\cite{Bhattacharya2018} for the minimum-load objective.
\begin{theorem}
    Given a set of $n$ points $P \subset \mathbb{R}^d$, parameters $k$ and $0 < \epsilon < 1$, there is a randomized algorithm that in time $2^{\tilde{O}(k/\epsilon^{O(1)})} nd$  outputs a list $L$ of $2^{\tilde{O}(k/\epsilon^{O(1)})}$ sets of centers of size $k$ such that for any partition $\mathbb{P}^* = \{P_1^*$, \ldots, $P_k^*\}$ of $P$ the following event occurs with probability at least $1/2$: there is a set $C$  in $L$ such that 
    \[cost_C(\mathbb{P}^*) \le (1 + \epsilon) \max_{i \in [k]} \cost (P_i^*).\]
    \label{thm:euclidean_centers}
\end{theorem}
\begin{proof}
    The algorithm proceeds exactly as Algorithm 5.1 in~\cite{Bhattacharya2018}. For the analysis, we observe that Bhattacharya et al. prove the following statement (follows immediately from invariant $P(i)$ in~\cite{Bhattacharya2018}): With constant probability, there is a set of centers $C = \{c_1, \ldots, c_k\}$ in the output of the algorithm and the permutation $i_1$, \ldots, $i_k$ of the clusters in $\mathbb{P}^*$ such that for each $j \in [k]$, 
    \[\cost_{c_j} (P_{i_j}^*) \le (1 + \frac{\epsilon}{2}) \cdot \cost(P_{i_j}^*) + \frac{\epsilon}{2k} \cdot \sum_{i = 1}^k \cost(P_i^*).\]
    From here it easily follows that the set of centers $C$ achieves $(1 + \epsilon)$-approximation of the cost of $\mathbb{P}^*$ with respect to the minimum-load objective:
    \[\cost_C(\mathbb{P}^*) \le \max_{j \in [k]} \cost_{c_j} (P_{i_j}^*) \le (1 + \frac{\epsilon}{2}) \cdot \max_{j \in [k]} \cost(P_{i_j}^*) + \frac{\epsilon}{2k} \cdot \sum_{i = 1}^k \cost(P_i^*) \le (1 + \epsilon) \cdot \max_{j \in [k]} \cost(P_{i_j}^*),\]
    since $\sum_{i = 1}^k \cost(P_{i}^*) \le k \max_{j \in [k]} \cost(P_{i_j}^*)$.
\end{proof}

In the general metric case, a similar result can be shown, however with the approximation factor of $(3 + \epsilon)$. Specifically, we show an analogue of Theorem~5 in~\cite{GoyalJK20} for the minimum-load objective. Similarly to Theorem~\ref{thm:euclidean_centers}, the algorithm and the analysis is identical to what is presented in~\cite{GoyalJK20}, up to a different view on the cost upper bound.

\begin{theorem}
    Given a set of $n$ points $P $ in a metric space, parameters $k$ and $0 < \epsilon < 1$, there is a randomized algorithm that in time $(k/\epsilon)^{O(k)} n$  outputs a list $L$ of $(k/\epsilon)^{O(k)}$ sets of centers of size $k$ such that for any partition $\mathbb{P}^* = \{P_1^*$, \ldots, $P_k^*\}$ of $P$ the following event occurs with probability at least $1/2$: there is a set $C$ in $L$ such that 
    \[\cost_C(\mathbb{P}^*) \le (3 + \epsilon) \max_{i \in [k]} \cost (P_i^*).\]
    \label{thm:metric_centers}
\end{theorem}
\begin{proof}
    The algorithm proceeds exactly as Algorithm 1 in~\cite{GoyalJK20}. For the analysis, we observe that Goyal et al. prove the following statement (encapsulated by \textbf{Property-I} in~\cite{GoyalJK20}): With constant probability, there is a set of centers $C = \{c_1, \ldots, c_k\}$ in the output of the algorithm and the permutation $i_1$, \ldots, $i_k$ of the clusters in $\mathbb{P}^*$ such that for each $j \in [k]$, 
    \[\cost_{c_j} (P_{i_j}^*) \le (3 + \frac{\epsilon}{2}) \cdot \cost(P_{i_j}^*) + \frac{\epsilon}{2k} \cdot \sum_{i = 1}^k \cost(P_i^*).\]
    Now, analogously to the proof of Theorem~\ref{thm:euclidean_centers}, it follows that the set of centers $C$ achieves $(3 + \epsilon)$-approximation of the cost of $\mathbb{P}^*$ with respect to the minimum-load objective:
    \[\cost_C(\mathbb{P}^*) \le \max_{j \in [k]} \cost_{c_j} (P_{i_j}^*) \le (3 + \frac{\epsilon}{2}) \cdot \max_{j \in [k]} \cost(P_{i_j}^*) + \frac{\epsilon}{2k} \cdot \sum_{i = 1}^k \cost(P_i^*) \le (3 + \epsilon) \cdot \max_{j \in [k]} \cost(P_{i_j}^*),\]
    since $\sum_{i = 1}^k \cost(P_{i}^*) \le k \max_{j \in [k]} \cost(P_{i_j}^*)$.
\end{proof}

Now, Theorem~\ref{thm:euclidean_centers} and Theorem~\ref{thm:metric_centers} imply that for the Minimum-Load $k$-Clustering problem with any given set of constraints on the desired clustering, there exists a $(1 + \epsilon)$-approximation algorithm in the Euclidean case, and a $(3 + \epsilon)$-approximation algorithm in the general metric case. The running time is $2^{\tilde{O}(k/\epsilon^{O(1)})}(nd + T)$ for both algorithms, where $T$ is the running time of an algorithm solving the respective assignment problem, either exact or $(1 + \epsilon)$-approximate.
In particular, combining the theorems with our approximation algorithm for $(\alpha, \beta)$-fair assignment (Theorem~\ref{thm:fairassgn}), for the FMLkC problem we obtain a $(1 + \epsilon)$-approximation in $\mathbb{R}^d$ and a $(3 + \epsilon)$-approximation in general metric in FPT time when parameterized by the number of clusters $k$ and the number of protected groups $\ell$.

\begin{theorem}
    For any $0 < \epsilon < 1$, there exists a randomized $(1 + \epsilon)$-approximation algorithm for $(\alpha, \beta)$-Fair Minimum-Load $k$-Clustering in $\mathbb{R}^d$ with running time
    \[2^{\tilde{O}(k \ell^2/\epsilon^{O(1)})} n^{O(1)}d.\]
    The same holds in general metric with the approximation factor of $(3 + \epsilon)$, where the running time becomes
    \[2^{\tilde{O}(k \ell^2/\epsilon)} n^{O(1)}.\]
    \label{thm:fpt_approx}
\end{theorem}
\begin{proof}
    First, we deal with the Euclidean case. Fix an optimal fair min-load $k$-clustering $\mathbb{P}^* = \{P_1^*$, \ldots, $P_k^*\}$ of $P$. Run the algorithm of Theorem~\ref{thm:euclidean_centers} on $P$ with error parameter $\epsilon_0$ to obtain the list $L$ of candidate sets of centers, here $\epsilon_0$ is such that $(1 + \epsilon_0)^2 \le (1 + \epsilon)$. In the following, assume that the event described in the statement of Theorem~\ref{thm:euclidean_centers} occurs for the clustering $\mathbb{P}^*$, by constant number of repetitions the probability of this can be lifted arbitrarily close to one. That is, there exists a set of $k$ centers $C'$ in $L$ such that
    \begin{equation}
    cost_{C'}(\mathbb{P}^*) \le (1 + \epsilon_0) \max_{i \in [k]} \cost (P_i^*).
        \label{eq:best_centers}
    \end{equation}
    Now, for each set of centers $C$ in $L$, run the $(1 + \epsilon_0)$-approximate assignment algorithm given by Theorem~\ref{thm:fairassgn} on $(P, C)$, and choose the set of centers $C''$ that gives the best assignment cost among the considered sets, denote the computed assignment from $P$ to $C''$ by $\varphi$. In what follows, we show that the set of centers $C''$ and the assignment $\varphi: P \to C''$ provide $(1 + \epsilon)$-approximate solution to the given FMLkC instance. Denote by $\psi$  the assignment from $P$ to $C'$ that the algorithm outputs,
    \[\cost(\varphi) \le \cost(\psi) \le (1 + \epsilon_0) \cost_{C'} (\mathbb{P}^*) \le (1 + \epsilon_0)^2 \max_{i \in [k]} \cost (P_i^*) \le (1 + \epsilon) \max_{i \in [k]} \cost (P_i^*),\]
    where the first inequality is by the choice of $C''$ and $\varphi$, the second is by Theorem~\ref{thm:fairassgn}, and the third inequality is by \eqref{eq:best_centers}.

    Finally, we show that the running time bound holds. Invoking Theorem~\ref{thm:euclidean_centers} takes time $2^{\tilde{O}(k/\epsilon_0^{O(1)})} nd$, and produces a list of $2^{\tilde{O}(k/\epsilon_0^{O(1)})}$ sets of centers. On each of them, running the algorithm of Theorem~\ref{thm:fairassgn} takes time
$k^{O(k\ell)} {\ell}^{O((k\ell^2/\epsilon_0)\log (\ell/\epsilon_0))} n^{O(1)}d$. Since $\epsilon_0 = O(\epsilon)$, the total running time can be bounded as
\[2^{\tilde{O}(k/\epsilon^{O(1)})} \left(nd + k^{O(k\ell)} {\ell}^{O((k\ell^2/\epsilon)\log (\ell/\epsilon))} n^{O(1)}d\right) = 2^{\tilde{O}(k \ell^2/\epsilon^{O(1)})} n^{O(1)}d.\]

    The general metric case is identical, but to obtain the list of candidate sets of centers we use Theorem~\ref{thm:metric_centers} instead of Theorem~\ref{thm:euclidean_centers}. The final cost bound changes to
    \[\cost(\varphi) \le \cost(\psi) \le (1 + \epsilon_0) \cost_{C'} (\mathbb{P}^*) \le (3 + \epsilon_0) \cdot (1 + \epsilon_0) \max_{i \in [k]} \cost (P_i^*) \le (3 + \epsilon) \max_{i \in [k]} \cost (P_i^*),\]
    where $\epsilon_0$ is chosen so that $(3 + \epsilon_0)\cdot (1 + \epsilon_0) \le (3 + \epsilon)$.
\end{proof}

\bibliography{clustering-bib}


\end{document}